\documentclass[10pt, conference, a4paper]{IEEEtran}
\usepackage{amsmath,amsfonts,amssymb,amsbsy, amsthm,ae,aecompl}
\usepackage{algorithm,algpseudocode}
\usepackage[utf8]{inputenc}
\usepackage[english]{babel}
\usepackage{bm}
\usepackage{color}
\usepackage[noadjust]{cite}
\usepackage{epsfig}
\usepackage{enumerate}
\usepackage{float}
\usepackage{fancyhdr}
\usepackage[T1]{fontenc}
\usepackage[acronym,toc,shortcuts]{glossaries}
\usepackage{graphicx, caption}
\usepackage{lipsum}
\usepackage{dsfont}
\usepackage{subfig}
\usepackage{transparent}
\usepackage{tikz,pgfplots}
\usepackage{times}
\usepackage{verbatim}
\usepackage[all]{xy}
\usepackage{epstopdf}

\usetikzlibrary{calc}
\usetikzlibrary{arrows,decorations.markings}
\pgfplotsset{
  grid style = {
    dash pattern = on 0.025mm off 0.95mm on 0.025mm off 0mm, 
    line cap = round,
    black,
    line width = 0.5pt
  },
  tick label style={font=\small},
  label style={font=\small},
  legend style={font=\footnotesize},
}

\pgfplotscreateplotcyclelist{laneas-delay1}{
	cyan!80!black, solid, thick, mark=none\\
	red!80!black, solid, thick, mark=none\\
	lime!40!black, solid, very thick, mark=none\\
	lime!60!black, solid, very thick, mark=none\\
	lime!80!black, solid, very thick, mark=none\\
	cyan!80!black, solid, thick, only marks, mark=+, mark size=2, every mark/.append style={solid,fill=cyan!80!black}\\
	red!80!black, solid, thick, only marks, mark=star, mark size=2, every mark/.append style={solid,fill=red!80!black}\\
	lime!40!black, solid, thick, only marks, mark=square, mark size=2, every mark/.append style={solid,fill=lime!80!black}\\
	lime!60!black, solid, thick, only marks, mark=o, mark size=2, every mark/.append style={solid,fill=lime!80!black}\\									
	lime!80!black, solid, thick, only marks, mark=diamond, mark size=2, every mark/.append style={solid,fill=lime!80!black}\\
}

\graphicspath{{figures/}}

\bgroup
\def\arraystretch{1.0}%

\newtheorem{definition}{Definition}

\newtheorem{theorem}{Theorem}

\newtheorem{proposition}{Proposition}
\newtheorem{corollary}{Corollary}
\newtheorem{remark}{Remark}

\makeglossaries
\newacronym{CP}{CP}{content provider}
\newacronym{D2D}{D2D}{device-to-device}
\newacronym{HetNet}{HetNet}{heterogeneous network}
\newacronym{MBS}{MBS}{macro base station}
\newacronym{MNO}{MNO}{mobile network operator}
\newacronym{SBS}{SBS}{small base station}
\newacronym{UE}{UE}{user equipment}
\newacronym{QoS}{QoS}{quality-of-service}
\newacronym{QoE}{QoE}{quality-of-experience}
\newacronym{GNE}{GNE}{generalized Nash equilibrium}
\newacronym{NE}{NE}{Nash equilibrium}
\newacronym{BR}{BR}{best response}
\begin{document}
\title{A Stackelberg Game for Incentive Proactive Caching Mechanisms in Wireless Networks \vspace{-0.3cm}}
\author{
		\IEEEauthorblockN{Fei Shen$^{\diamond}$, Kenza Hamidouche$^{\diamond}$, Ejder Baştuğ$^{\dagger, \diamond}$ and Mérouane Debbah$^{\diamond, \circ}$\\ }
		\IEEEauthorblockA{
				\vspace{-0.1cm}
				\small
				$^{\diamond}$Large Networks and Systems Group (LANEAS), CentraleSupélec, \\ Université Paris-Saclay, 3 rue Joliot-Curie,  91192 Gif-sur-Yvette, France \\	
				$^{\dagger}$Research Laboratory of Electronics, Massachusetts Institute of Technology, \\ 77 Massachusetts Avenue, Cambridge, MA 02139, USA \\
				$^{\circ}$Mathematical and Algorithmic Sciences Lab, Huawei France R\&D, Paris, France \\
				\{fei.shen, kenza.hamidouche\}@centralesupelec.fr, ejder@mit.edu, merouane.debbah@huawei.com	
\vspace{-1.0cm}		}
		\thanks{This research has been supported by the ERC Starting Grant 305123 MORE (Advanced Mathematical Tools for Complex Network Engineering), the  U.S.  NSF Grant CCF-1409228, and the projects 4GinVitro and BESTCOM.}
}
\IEEEoverridecommandlockouts
\maketitle
\begin{abstract}
In this paper, an incentive proactive cache mechanism in cache-enabled small cell networks (SCNs) is proposed, in order to motivate the content providers (CPs) to participate in the caching procedure. A network composed of a single mobile network operator (MNO) and multiple CPs is considered. The MNO aims to define the price it charges the CPs to maximize its revenue while the CPs compete to determine the number of files they cache at the MNO's small base stations (SBSs) to improve the quality of service (QoS) of their users. This problem is formulated as a Stackelberg game where a single MNO is considered as the leader and the multiple CPs willing to cache files are the followers. The followers game is modeled as a non-cooperative game and both the existence and uniqueness of a Nash equilibrium (NE) are proved. The closed-form expression of the NE which corresponds to the amount of storage each CP requests from the MNO is derived. An optimization problem is formulated at the MNO side to determine the optimal price that the MNO should charge the CPs. Simulation results show that at the equilibrium, the MNO and CPs can all achieve a utility that is up to $50$\% higher than the cases in which the prices and storage quantities are requested arbitrarily.
\end{abstract}
\vspace{-0.1cm}
\section{Introduction}
The ever increasing number of mobile phones and connected devices is expected to contribute to an $800$\% increase in mobile data traffic in the five upcoming years \cite{Cisco}. Bandwidth intensive applications such as video-on-demand  traffic will represent more than $70$\% of the global generated data requests \cite{Cisco}. To support this growing traffic and offload macro base stations, short range \glspl{SBS} are deployed closer to the end-users. However, these \glspl{SBS} are connected to a core network via capacity-limited backhaul links which makes it difficult to meet users' requirements in terms of \ac{QoS}, especially during peak hours. To deal with this problem, distributed caching at the network edge has recently been proposed as a promising solution \cite{Bastug2014LivingOnTheEdge}.

The idea of distributed caching consists in equipping the \glspl{SBS} with storage units in which files are cached according to a placement policy. Thus, the \glspl{SBS} can serve most of the requests locally without using the backhaul. However, for a successful deployment of proactive edge caching, the \glspl{MNO} require cooperation of the \glspl{CP} to be able to cache their content at the \glspl{SBS} \cite{Pachos2016Technical}. To this end, incentive mechanisms must be developed by the operators to incite \glspl{CP} to share and cache their content. The \ac{MNO} offers to the \glspl{CP} the caching service that allows the \glspl{CP}' users to improve their \ac{QoS} and in return the \glspl{CP} pay a price defined by the \ac{MNO} depending on the amount of storage space requested by each \ac{CP}.

Recently, several caching works have appeared from different aspects, such as optimal caching policies in layered video delivery \cite{Poularakis2016Caching}, \ac{D2D} networks \cite{Gregori2016Wireless}, hieararchical caching \cite{Ghoreishi2016Provisioning}, multi-cell scenario with limited-backhaul \cite{Peng2016Cache} and so on (see \cite{Pachos2016Technical} for detailed discussions). More relevantly, there exist some works focusing on economic aspect of caching in wireless \ac{D2D} networks. In such networks, the operators define pricing scheme to motivate  users to proactively download the most popular files and cache them in their devices to serve other users' requests. In \cite{Alotaibi2015Towards}, the authors proposed a smart pricing scheme to maximize the benefit of the operator and minimize the charged price to the users. Via \ac{D2D} communications, users can trade their cached files to minimize their expected payments. On the other hand, the operator defines a dynamic pricing model that differentiates off-peak and peak time periods to maximize its own benefit. The authors in \cite{Chen2016CachingStackelberg} formulated the cache incentive problem as a Stackelberg game in which the \glspl{SBS} are the leaders and the users are the followers. In this model, the \glspl{SBS} start by anticipating users' reactions and determine the optimal price that maximizes their offloaded traffic from to the users devices. Given a defined reward by the \glspl{SBS}, users can then decide whether to help the \glspl{SBS} by caching the files and serving other users or not.

Despite being interesting, all these works focused on cache incentive mechanisms in \ac{D2D} wireless networks and ignored the role of \glspl{CP} in the caching process. 
The main contribution of this work is to propose a new cache incentive mechanism between an \ac{MNO} and multiple \glspl{CP}. We formulate the cache incentive problem as a Stackleberg game in which the \ac{MNO} is a leader and the \glspl{CP} are the followers. The \ac{MNO} predicts the requests profile of the \glspl{CP} and define the price that maximizes its revenue. On the other hand, due to limited storage capacity of the \glspl{SBS}, \glspl{CP} cannot cache all their files, thus, the \glspl{CP} compete to maximize the amount of storage space they request given the fixed price by the \ac{MNO}. The competition between the \glspl{CP} is formulated as a non-cooperative sub-game in which each \ac{CP} aims to maximize the satisfaction of its users in terms of transmission rate. Both the existence and uniqueness of the \ac{NE} are proved. The \ac{NE} represents a state in which none of the \glspl{CP} can improve its benefit by requesting a different amount of storage space given the requested amount by the other \glspl{CP} fixed. We provide closed-form expressions of the storage amount the \glspl{CP} must request at the \ac{NE}. Then, given the request profiles of the \glspl{CP}, an optimization problem is formulated at the \ac{MNO} to determine the optimal price that should be charged to the \glspl{CP} to maximize the \ac{MNO}'s revenue. Simulation results show that the \ac{MNO} can get a revenue that is up to $50$\% higher compared to the case in which the prices are selected arbitrarily. Moreover, at the equilibrium, we show that the \glspl{CP} can achieve a utility that is $20$\% to $50$\% higher compared to the two cases in which half and the double of the storage space at the \ac{NE} are requested by the \glspl{CP}, respectively.

The rest of the paper is organized as follows. In Section \ref{sec:pre}, we define the system and the caching models. The Stackelberg game is formulated in Section \ref{sec:sgame}, and the analysis is conducted in Section \ref{sec:solution}. The numerical results are presented in Section \ref{sec:numerical}. We finally draw conclusions in Section \ref{sec:conclusions}.

\vspace{-0.1cm}
\section{Preliminaries}\label{sec:pre}
%
\vspace{-0.1cm}
\subsection{System  Model}
We consider an \ac{HetNet} which consists of $N$ cache-enabled \glspl{SBS} and is controlled by one \ac{MNO}. There are $M$ \glspl{CP} that are willing to cache their files in the \ac{HetNet} in order to enhance the \ac{QoS} of their users. Let $\mathcal{N} = \{1, 2, \dots, N\}$ denote the set of cache-enabled \glspl{SBS}, each of which has a limited capacity $s_n$, and $\mathcal{M} = \{1, 2, \dots, M\}$ denote the set of \glspl{CP}. Each \ac{CP} requests from the \ac{MNO} to cache its $q_m$ most popular files. The $N$ \glspl{SBS} are connected to the \ac{MNO} with a backhaul such as DSL, optical fibers or wireless backhaul. All the \glspl{UE} communicate with the corresponding \glspl{SBS}. If the desired files from the \glspl{UE} are not cached in the \glspl{SBS}, they are served by the \ac{MNO} via the capacity-limited backhaul links. The system model is depicted in Fig. \ref{fig:scenario}.

\begin{figure}[t]	
	\centering
	\includegraphics[width=0.9\linewidth]{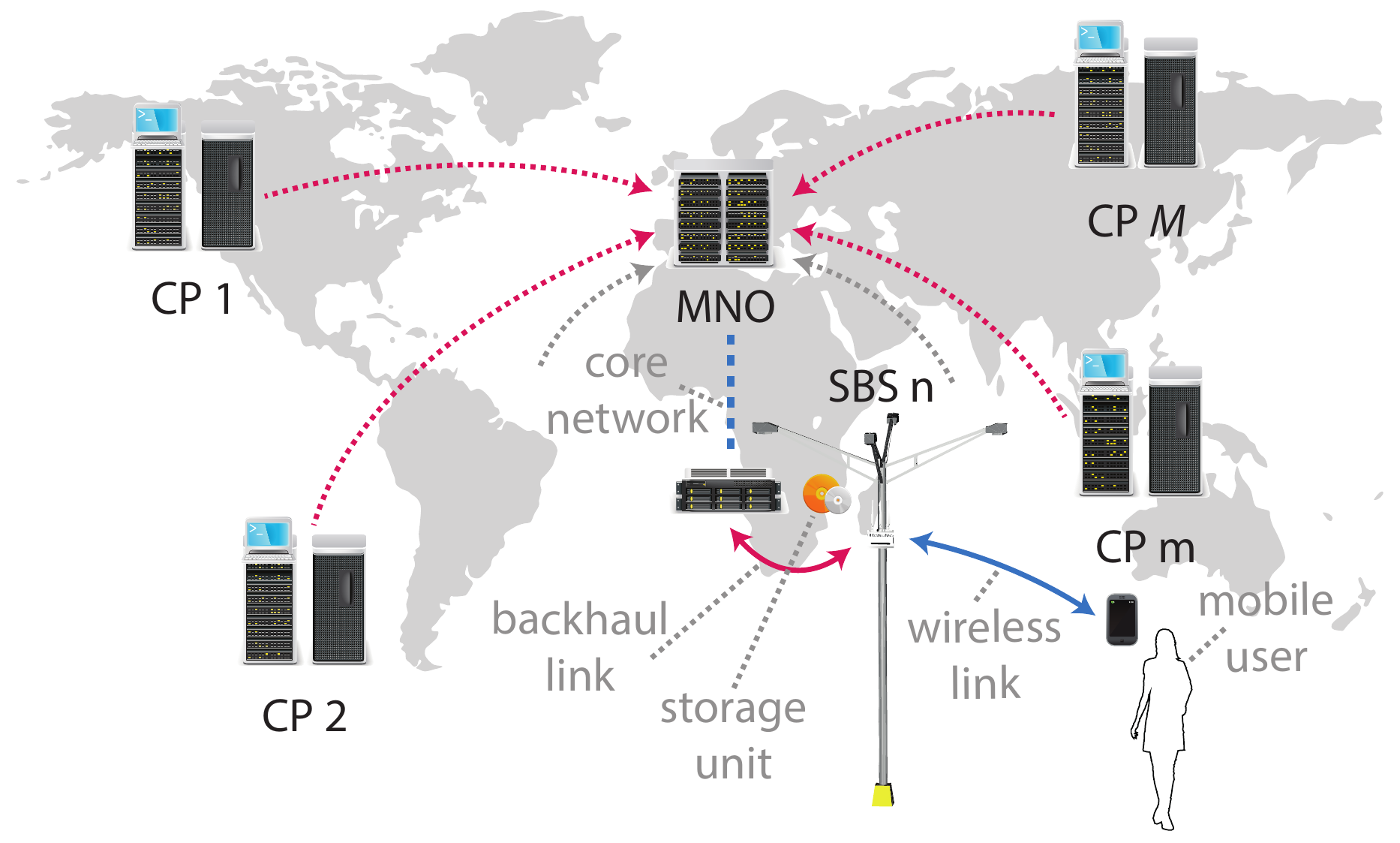}
	\caption{\small An illustration of the system model.
}
\vspace{-0.6cm}	\label{fig:scenario}
\end{figure}

Each \ac{CP} $m$ has a local content catalog $\mathcal{F}_m$ with $F_m$ files. The global files catalog is denoted $\mathcal{F} = \bigcup \mathcal{F}_m$, and all the files are assumed to have the same unit size. 
We assume that each file $f_i$ is cached by the MNO $p_{f_i}$ times at the SBSs based on the popularity of the file.  
The goal of each \ac{CP} is to cache its most popular files locally at the \glspl{SBS} so that the users experience an improved data transmission rate. However, the storage capacity of the \glspl{SBS} is limited and not all the files can be cached with sufficient copies. The notations used in the rest of this paper is summarized in Table \ref{table:notation}.

%
\vspace{-0.1cm}
\subsection{Cache Incentive Model}
We formulate the cache incentive problem as a Stackelberg game in which a non-cooperative game is involved as a sub-game. As shown in Fig. \ref{fig:scenario}, the Stackelberg game is played between the \ac{MNO} and the $M$ \glspl{CP}. On the one hand, all the \glspl{CP} wish to cache as many as possible files at the MNO's \ac{HetNet} such that the \ac{QoS} of their users (for example the delay) is improved. On the other hand, there exists a cost for caching. Therefore, the incentive proactive caching mechanism is controlled by a charge price determined by the MNO to optimize its revenue in the proposed Stackelberg game. The sequential game (Stackelberg) leads to a more competitive equilibrium than the simultaneous move game: The \ac{MNO} plays as the leader providing the caching price $\pi$ to all the \glspl{CP}, while the \glspl{CP} play as the followers reacting with their optimal number of files they want to cache based on the given price. Both the \ac{MNO} and the \glspl{CP} are rational and optimize their own utilities denoted as $u_o$ and $u_m$, respectively. The strategy of the leader \ac{MNO} is the caching price $\pi$ and the strategies of the followers \glspl{CP} are their the number of files they want to cache $q_m(\pi)$. It should be noted that the price defined by the MNO depends on the used caching policy and the storage capacity of all the SBSs.

Due to the limited caching capacity, the $M$ \glspl{CP} compete for the possible caching quantities. Therefore, a non-cooperative game is played as a sub-game among the $M$ \glspl{CP}. We assume that the file access probability of each file is perfectly known at all the entities (i.e., \glspl{CP} and \ac{MNO}). The strategies of the \glspl{CP} are the quantity of their caching requests. Under the perfect knowledge of the access probability of each caching file, the \glspl{CP} will choose to cache the first $q_m^*$ files ordered with the highest access probability, i.e., popularity.

\bgroup
\def\arraystretch{1.7}%
\begin{table}
\scriptsize
\caption{List of Notations. \label{table:notation}}	
\begin{tabular}{ | l | l | p{2cm}|}
\hline
$\mathcal{N} = \{1, 2, \dots, N\}$ & set of \glspl{SBS} \\
\hline
$\mathcal{M} = \{1, 2, \dots, M\}$ & set of \glspl{CP}\\
\hline
$q_m$ & The number of files that \ac{CP} $m$ wants to cache  \\
\hline
$p_{f_i}$ & The number of cached copies of file $f_i$ \\
\hline
$s_n$ & The cache size of each SBS $n$ \\
\hline
$\pi$ & The charged price for caching a given file \\
\hline
$u_o$ & The utility of \ac{MNO} \\
\hline
$u_m$ & The utility of \glspl{CP}\\
\hline
$\mathcal{Q}_m$ & The set of files requested by each \ac{CP} $m$\\
\hline
$p_m$ & The mean access probability of \ac{CP} $m$'s files\\
\hline
$d(\mathcal{Q}_m)$ & The total caching copies of requested files of all \glspl{CP}\\
\hline
\end{tabular}
\vspace{-0.4cm}
\end{table}
\bgroup
\def\arraystretch{1.0}%
%
\section{Stackelberg Game Formulation}\label{sec:sgame}
In this section, we provide the utilities of the \ac{MNO} and the \glspl{CP}, respectively. The proposed Stackelberg game is played as follows. The leader plays first by providing its optimal strategy to all the followers after predicting the strategies of the followers. Then, the followers reply with their best strategies given the strategy of the leader. The result of the Stackelberg game is that both the leader and the followers optimize their utilities. The Stackelberg equilibrium is exactly the point that the leader wishes. The basic idea of the utilities for both the \ac{MNO} and the \ac{CP} can be formulated as a general function, such as $Utility = Revenue - Cost$. In the following, we will analyse the utilities of the \ac{MNO} and the \glspl{CP}, respectively.
\subsection{Utility of the \ac{MNO}}\label{sec:sgame:mo}
For the \ac{MNO}, the main cost is the caching storage cost denoted as $C_o$. This caching storage cost is a function of the caching capacity of all the SBSs and the number of times each file is cached. Therefore, the cost of caching all CPs files $C_o$ for the MNO can be defined as the following barrier function:
\begin{eqnarray}\label{c1}
{\small
C_o = \left\{ \begin{array}{ll}
\frac{1}{S - d(\mathcal{Q_M})} & \textrm{if $0 < d(\mathcal{Q_M}) \le S$}\\
\infty & \textrm{otherwise,}
\end{array} \right.
}
\end{eqnarray}
where $d(\mathcal{Q_M})$ is the amount of all the cached files by the MNO and is given by:
\begin{eqnarray}\label{aq}
d(\mathcal{Q_M}) = \sum_{m=1}^M \sum_{f_i \in \mathcal{Q}_m} p_{f_i},
\end{eqnarray}
where $p_{f_i}$ is the number of copies of file $f_i$ that are cached at the SBSs and defined by the \ac{MNO} depending on the popularity of file $f_i$.
The caching capacity of the MNO is given by:
\begin{eqnarray}
S = \sum_{n=1}^N s_n.
\end{eqnarray}

The revenue of the \ac{MNO} in the caching problem is the total charge from all the \glspl{CP} for caching their files. The price of caching one file is denoted by $\pi$ and the number of files that is requested by a \ac{CP} $m$ to be cached is $q_m$. Thus, the total revenue of the \ac{MNO} can be given by
\begin{eqnarray}\label{r1}
R_o = \pi \sum_{m=1}^M q_m(\pi).
\end{eqnarray}

Now we obtain the utility function of the \ac{MNO} in our proposed Stackelberg game as a function of the quantity of caching request from all the \glspl{CP}, $\mathbf{q} = [q_1, \dots, q_M]$.
\begin{eqnarray}
u_o(\mathbf{q}(\pi)) = R_o(\mathbf{q}(\pi)) - C_o(\mathbf{q}(\pi)) ,
\end{eqnarray}
where $R_o$ and $C_o$ are defined in (\ref{c1}) and (\ref{r1}), respectively.

%
\subsection{Utility of the CPs}
\label{sec:sgame:cp}
Given that the \ac{MNO} fixes a price $\pi$ for a storage unit, all the \glspl{CP} reply with their quantity of caching request by optimizing their own utilities. We assume that each \ac{CP} is independent of the others and makes its best strategy only based on its local information.

The cost of \glspl{CP} for the caching requests are the charge paid to the \ac{MNO} for their desired quantity of caching files. For each \ac{CP} $m$, the cost $C_m$ is defined as
\begin{eqnarray} \label{c4}
C_m = \pi q_m(\pi),
\end{eqnarray}
where $q_m(\pi)$ is a function of the provided charge price $\pi$. $C_m$ is an increasing function of the caching quantity $q_m(\pi)$.

In fact, the \glspl{CP} such as Facebook and Youtube are concerned about the satisfaction of their users. This user satisfaction, which can be the delay of downloading a file or the data throughput, is an increasing function of the caching request quantity $q_m$ of \ac{CP} $m$ and a decreasing function of the caching request quantities $\mathbf{q}_{-m}$ of others because of the limited caching capacity. $\mathbf{q}_{-m}$ is defined as the caching request quantities of all the other \glspl{CP} except \ac{CP} $m$, i.e.,
\begin{eqnarray}
\mathbf{q}_{-m} = [q_{1}, \dots, q_{m-1}, q_{m+1}, \dots, q_{M}].
\end{eqnarray}

The revenue of \glspl{CP} can be easily defined as the satisfaction function of each \ac{CP}. We assume that the $q_m$ files with the highest access probabilities in each \ac{CP} $m$  are requested. Then the satisfaction function of \ac{CP} $m$ is as follows, which is an increasing function of the caching request quantity.
\begin{eqnarray}\label{r2}
R_m = \log \left(1+ \frac{q_m}{1+ \frac{1}{\alpha_m} J_m} \right),
\end{eqnarray}
where $J_m = \sum_{l\ne m}{q}_{l}$ is the quantity of caching files requested by all the other \glspl{CP} except \ac{CP} $m$ and $\alpha_m$ is the amount of generated requests by the users of CP $m$. The satisfaction $R_m$ of a CP $m$ was chosen as it is an increasing function of the number of files that are cached by a CP $m$ as well as its traffic load $\alpha_m$. Indeed, the more files are cached at the SBSs and the higher the popularity of the files, the higher is the number of requests that can be served locally from the SBSs. Thus, more users of CP $m$ can experience a higher transmission rate. On the other hand, the satisfaction of a CP $m$ decreases as the number of files cached by the other CPs increases. In this case, less storage is available for CP $m$ and a higher number of its requests need to be served through the backhaul. It should be noted that any other function that has the same properties as $R_m$ can be chosen as a satisfaction function.

Given the constraints on the caching capacity of the mobile network, now we obtain the utility function of each \ac{CP} $m$ as
\begin{eqnarray}\label{eq:1}
u_m^{} = R_m - C_m,
\end{eqnarray}
where the cost $C_m$ and the revenue $R_m$ are defined in (\ref{c4}) and (\ref{r2}), respectively. It is obvious that $u_m^{}$ is a concave function of the quantity of the caching requests.

Given the charge price from the \ac{MNO}, all the \glspl{CP} respond with their quantities of caching requests simultaneously by maximizing their own utilities. Due to the limited caching capacity $S$, all the \glspl{CP} can be considered as competitive players competing for the caching quantities. Therefore, we can formulate the quantity optimization problem of \glspl{CP} as a non-cooperative game $\mathcal{G} = \mathcal{G}(\mathcal{M}, \mathcal{Q}, \mathcal{U})$ consisting of the following components:
\begin{itemize}
\item The set of players in $\mathcal{G}$ is the set $\mathcal{M}$ of \glspl{CP}.
\item Given the quantity profile $\mathbf{q}_{-m} \in \mathcal{Q}_m = \prod_{l\ne m} \mathcal{Q}_l$ of the opponents of player $m$, the feasible action set of \ac{CP} in the presence of the caching capacity $d(\mathcal{Q}_M)\le S$ is
\begin{eqnarray}
\mathcal{Q}_M(\mathbf{q}_{-m}) = \{\mathbf{q}_m \in \mathcal{Q}_m : \mathbf{q}_m\ge 0\}.
\end{eqnarray}
\item The utility $u_m^{}$ of player $m$ is given by (\ref{eq:1}).
\end{itemize}

In this context, the most adopted solution concept is the \ac{NE}.
\begin{definition}
A quantity profile $\mathbf{q}^*$ is a \ac{NE} of the caching request quantity game $\mathcal{G}$ if
\begin{eqnarray}
\mathbf{q}^*_m \in \mathcal{Q}_m (\mathbf{q}^*_{-m}), & \forall m \in \mathcal{M}, \\
u_m(\mathbf{q}^*) \ge u_m(\mathbf{q}_m, \mathbf{q}^*_{-m}), & \forall \mathbf{q}_m \in \mathcal{Q}_m(\mathbf{q}^*_{-m}).
\end{eqnarray}
\end{definition}
Provided the utility functions for the \ac{MNO} and \glspl{CP}, the optimal strategies of the leader and the followers are derived in the next section. In the Stackelberg game, the followers' strategies are predicted before the leader makes its own strategy. Therefore, the optimal quantity of caching requests of the \glspl{CP} is analyzed first.
%
\section{Stackelberg Game Solution}
\label{sec:solution}
In this section, we derive the solution of the proposed Stackelberg game. In the sequential (Stackelberg) game, the leader moves first by predicting the strategies of the followers. The followers reply by optimizing their own utilities given the strategy of the leader. The result of the game is that the followers play exactly what the leader wishes. Therefore, the optimal strategy of the followers should be predicted first. In the following, we obtain the optimal quantities of caching requests of the \glspl{CP} in the closed form.
%
\subsection{Optimal Quantity of Caching Request}
\label{sec:solution:oq}
The \glspl{CP} optimize their strategies of the quantity of the caching requests and provide them to the \ac{MNO} given the charge price $\pi$ announced by the \ac{MNO}. The optimal $q_m$ is obtained by solving the following problem.
\begin{equation}
q_m = \arg \max_{\mathbf{q}} u_m \quad \text{subject to} \quad \mathbf{q} \ge 0.
\end{equation}
Given the charge price for the caching files, the \glspl{CP} compete for the quantity of caching requests, which formulates a non-cooperative game. In the following, we find the solution for this non-cooperative game.
\begin{proposition}[Best Response]
Given the charge price $\pi$ announced by the \ac{MNO}, the \ac{BR} of each \ac{CP} is the quantity of caching files it requests, which is
\begin{eqnarray}\label{qm}
q_m^{BR} = \left(\frac{1}{\pi}-1- \frac{J_{m}}{\alpha_m}\right)^+,
\end{eqnarray}
with $\alpha_m \ge M$.
\end{proposition}
\begin{proof}
The \ac{BR} of the caching request quantity $q_m^{BR}$ is obtained by checking the first derivative of $u_m$ with respect to $q_m$,
\begin{eqnarray}
\frac{\partial u_m}{\partial q_m} = \frac{\alpha_m}{\alpha_m+\alpha_m q_m + J_{m}} - \pi = 0.
\end{eqnarray}
The second derivative of $u_m$ with respect to $q_m$ is
\begin{eqnarray}
\frac{\partial^2 u_m}{\partial q_m^2} = \frac{-\alpha_m^2}{(\alpha_m+\alpha_m q_m + J_{m})^2} < 0,
\end{eqnarray}
which guarantees a global optimal of $u_m$.
Since the utility of each \ac{CP} $u_m$ is a convex function of $q_m$, (\ref{qm}) is proved. The function $()^+$ is to ensure the requested caching quantity to be a non-negative value.
\end{proof}
\begin{remark}
The total quantity of the caching requests of all the other \glspl{CP} can be learnt or fed back at each \ac{CP}. The $q_m^{BR}$ is based on the local information $\alpha_m$, $q_m$ and the feedback $J_m$  as a single value of the summation.
\end{remark}
Now we prove that the proposed non-cooperative game admits a unique \ac{NE}.
\begin{theorem}[Nash Equilibrium]
Given the charge price $\pi$ announced by the \ac{MNO}, the \ac{NE} of each \ac{CP} is the quantity of caching files it requests, which is
\begin{eqnarray}\label{qmNE}
q_m^{NE} = \frac{D_m}{D} = \frac{(\frac{1}{\pi} - 1) a_m b_m}{D},
\end{eqnarray}
where $a_m$ is
\begin{eqnarray}\label{bm}
{\small
a_m = \left\{ \begin{array}{ll}
1- \frac{1}{\alpha_m} - \frac{1- \alpha_1}{\alpha_1} \sum_{l=2, l\ne m}^M \frac{\alpha_m - \alpha_l}{(\alpha_l - 1)\alpha_m}  & \textrm{if $m \ne 1$}\\
1- \frac{1}{\alpha_1} - \frac{1- \alpha_M}{\alpha_M} \sum_{l=2}^{M-1} \frac{\alpha_1 - \alpha_l}{(\alpha_l - 1)\alpha_1}  & \textrm{if $m = 1$}
\end{array} \right.
}
\end{eqnarray}
and $b_m$ is
\begin{eqnarray}\label{am}
b_m = \left\{ \begin{array}{ll}
\prod_{l=2, l\ne m}^M (1- \frac{1}{\alpha_l}) & \textrm{if $m \ne 1$}\\
\prod_{l=2}^{M-1} (1- \frac{1}{\alpha_l}) & \textrm{if $m = 1$}
\end{array} \right.
\end{eqnarray}
and
\begin{eqnarray}\label{D}
D = \left( 1 - \frac{1- \alpha_1}{\alpha_1}\sum_{l=2}^M \frac{1}{\alpha_l -1} \right)\prod_{l=2}^{M} \Big(1- \frac{1}{\alpha_l}\Big).
\end{eqnarray}
If $M =2$, then $\sum_{l =2}^{M-1} \frac{\alpha_1 - \alpha_l}{(\alpha_l -1) \alpha_1} = 0$ and $\prod_{l =0}^{M-1} 1- \frac{1}{\alpha_l} = 1$.
\end{theorem}
\begin{proof}
We give only the sketch of the proof due to the space limitations. The trick for finding the \ac{NE} of the proposed non-cooperative game is to jointly solve the $M$ functions of the \ac{BR} for all the \glspl{CP}. The $M$ functions of the \ac{BR} can be formulated as a matrix function denoted as $\mathbf{D} \mathbf{q} = \mathbf{C}$ where the matrix $\mathbf{D}$ and vectors $\mathbf{q}$ and $\mathbf{C}$ are as follows, respectively.
\begin{eqnarray}
\left[ \begin{array}{cccc}
1 & \frac{1}{\alpha_1} & \ldots & \frac{1}{\alpha_1} \\
\frac{1}{\alpha_2}  &1 & \ldots & \frac{1}{\alpha_2}\\
\vdots & \vdots & \ddots & \vdots\\
\frac{1}{\alpha_M} & \frac{1}{\alpha_M}& \ldots&1
\end{array} \right]
\left[ \begin{array}{c}
q_1\\
q_2\\
\vdots\\
q_M
\end{array} \right]
 =
\left[ \begin{array}{c}
\frac{1}{\pi}-1\\
\frac{1}{\pi}-1\\
\vdots\\
\frac{1}{\pi}-1
\end{array} \right].
\end{eqnarray}
The \ac{NE} quantity of each \ac{CP} is solved by applying the Cramer's rule $q_m = \frac{\text{det}(\mathbf{D}_m)}{\text{det}(\mathbf{D})}$, where $\text{det}(\mathbf{D})$ is the determinant of matrix $\mathbf{D}$ and $\text{det}(\mathbf{D}_m)$ is the determinant of matrix $\mathbf{D}_m$ which is formed by replacing the $m$-th column of $\mathbf{D}$ by the column vector $\mathbf{C}$.
\end{proof}
\begin{remark}
In caching, there are two basic modes. One is coded caching in which the SBSs cache any number of bits from the files. In this case, the \ac{NE} quantity derived in (\ref{qmNE}) is the actual quantity of files requested by \ac{CP} $m$. The other one is uncoded caching, in which a given file can only be cached as a whole or equally divided chunks. In this case, the round function $\lfloor \rceil$ is applied to $q_m^{NE}$ before \ac{CP} $m$ replies to the \ac{MNO} in order to obtain the integer optimal quantity of caching request. In our model, we assume the uncoded caching case so that only integer quantity of caching requests are considered.
\end{remark}
%
\subsection{Optimal Charge Price}
\label{sec:solution:ocp}
After predicting the strategy policies of the \glspl{CP}, the \ac{MNO} optimizes the charge price $\pi$ by solving the following problem.
\begin{equation}
\pi^* = \arg \max_{\pi} u_o(\mathbf{q}(\pi)) \quad \text{subject to} \quad {\pi} \ge 0.
\end{equation}
We observe from (\ref{c1}) and (\ref{aq}) that the quantity of caching request from each \ac{CP} $m$ exists in the size of the set $\mathcal{Q}_m$. Therefore, we assume that the number of copies of each file $f_i \in \mathcal{Q}_m$ that are cached are defined as the quantized vector given by:
\begin{eqnarray}\label{p}
\mathbf{p}_{f_i} = [\dots, p_m - \Delta p_m, p_m, p_m + \Delta p_m, \dots]
\end{eqnarray}
with size of $q_m$ for each \ac{CP} $m$. $p_m$ is the mean access probability of files requested by \ac{CP} $m$ and $\Delta p_m$ is the step size.

Given the number of times $\mathbf{p}_{f_i}$ each file $f_i$ is cached as defined in (\ref{p}), we obtain
\begin{eqnarray}\label{pp}
\sum_{f_i \in \mathcal{Q}_m}{p}_{f_i} = q_m^{NE} f(p_m),
\end{eqnarray}
where $f(p_m)$ is a function of the mean access probability $p_m$ with the following form
\begin{eqnarray}\label{pm}
{\small
f(p_m) = \left\{ \begin{array}{ll}
p_m & \textrm{when $q_m$ is odd,}\\
p_m + \frac{\Delta p_m}{2} & \textrm{when $q_m$ is even.}
\end{array} \right.
}
\end{eqnarray}
Since the \ac{MNO} can predict the \ac{NE} quantities of all the \glspl{CP}, which are functions of the charge price $\pi$, the utility function of the \ac{MNO} is then
\begin{eqnarray}
u_o = \pi \sum_{m =1}^M q_m(\pi) - \frac{1}{S- \sum_{m =1}^M q_m(\pi) f(p_m)}.
\end{eqnarray}
By predicting the requesting $q_m(\pi)$ from each \ac{CP} $m$, the \ac{MNO} can make its own optimal strategy, which is the charge price $\pi^*$. Here we analyse the case where the total caching requests do not exceed the caching capacity of the \ac{HetNet}.
\begin{proposition}[Optimal Price]
The optimal charge price $\pi^*$ provided by the \ac{MNO} to maximize its own utility $u_o$ is as follows.
\begin{eqnarray}\label{optiprice}
{\small
\pi^* = \frac{\sqrt{\frac{r}{t}} + r}{S+r},
}
\end{eqnarray}
where $r$ is a function of $\mathbf{\alpha}$ and $t$ is a function of both $\mathbf{\alpha}$ and $\mathbf{p}$ defined as
\begin{eqnarray}
t = \sum_{m=1}^M \frac{a_m b_m}{D},
\end{eqnarray}
and
\begin{eqnarray}
r = \sum_{m=1}^M \frac{a_m b_m f(p_m)}{D}.
\end{eqnarray}
\end{proposition}
\begin{proof}
Since $r$ and $t$ are independent of the optimization objective $\pi$, we rewrite the utility function of \ac{MNO} $u_o$ as
\begin{eqnarray}\label{uo}
u_o = (1-\pi) t - \frac{1}{S- (\frac{1}{\pi} - 1) r}.
\end{eqnarray}
Then the first derivative with respect to $\pi$ is
\begin{eqnarray}\label{price1}
\frac{\partial u_o}{\partial \pi} = \frac{r}{(S\pi - (1-\pi)r)^2} - t.
\end{eqnarray}
The second derivative with respect to $\pi$ is
\begin{eqnarray}
\frac{\partial^2 u_o}{\partial \pi^2} = \frac{-2r(S\pi - (1-\pi)r)(S+r)}{(S\pi - (1-\pi)r)^4}.
\end{eqnarray}
By observing that the \ac{MNO} will definitely provide a price in order to ensure $d(\mathcal{Q}_m) \le S$, $S\pi - (1-\pi)r > 0$ and then $\frac{\partial^2 u_o}{\partial \pi^2} <0$, which guarantees a global maximum $\pi^*$.

By solving $\frac{\partial u_o}{\partial \pi} =0$ in (\ref{price1}), we get the value of $\pi^*$ in (\ref{optiprice}), which completes the proof.

Notice that the charge price should be a positive value $\pi^*>0$, therefore the result $\frac{-\sqrt{\frac{r}{t}} + r}{S+r}$ is dropped.
\end{proof}
To ensure a positive total charge and also that the total caching requests do not exceed the caching capacity $S$, i.e., the cost $C_o$ is a positive limited value, the charge price given by the \ac{MNO} should be restricted in the following range.
\begin{corollary}\label{priceregion}
The caching capacity $S$ of the \ac{HetNet} can be fully exploited if the charge price given by the \ac{MNO} follows the range
\begin{equation}\label{pricerange}
\frac{r}{S+r} <\pi < 1.
\end{equation}
\end{corollary}
\begin{proof}
$\pi < 1$ is proved by ensuring $\frac{1}{\pi} - 1 >0$ in (\ref{uo}). $\frac{r}{S+r} < \pi$ is proved by ensuring $\frac{1}{S - (\frac{1}{\pi} -1)r} > 0$ in (\ref{uo}).
\end{proof}

\begin{remark}
The optimal price $\pi^*$ provided in (\ref{optiprice}) always satisfies the feasible price range in (\ref{pricerange}).
\end{remark}
\begin{figure*}[!ht]
\centering
\scalebox{0.95}{
\subfloat[ \label{fig:convergence-3}]{%
	\includegraphics[width=0.33\linewidth]{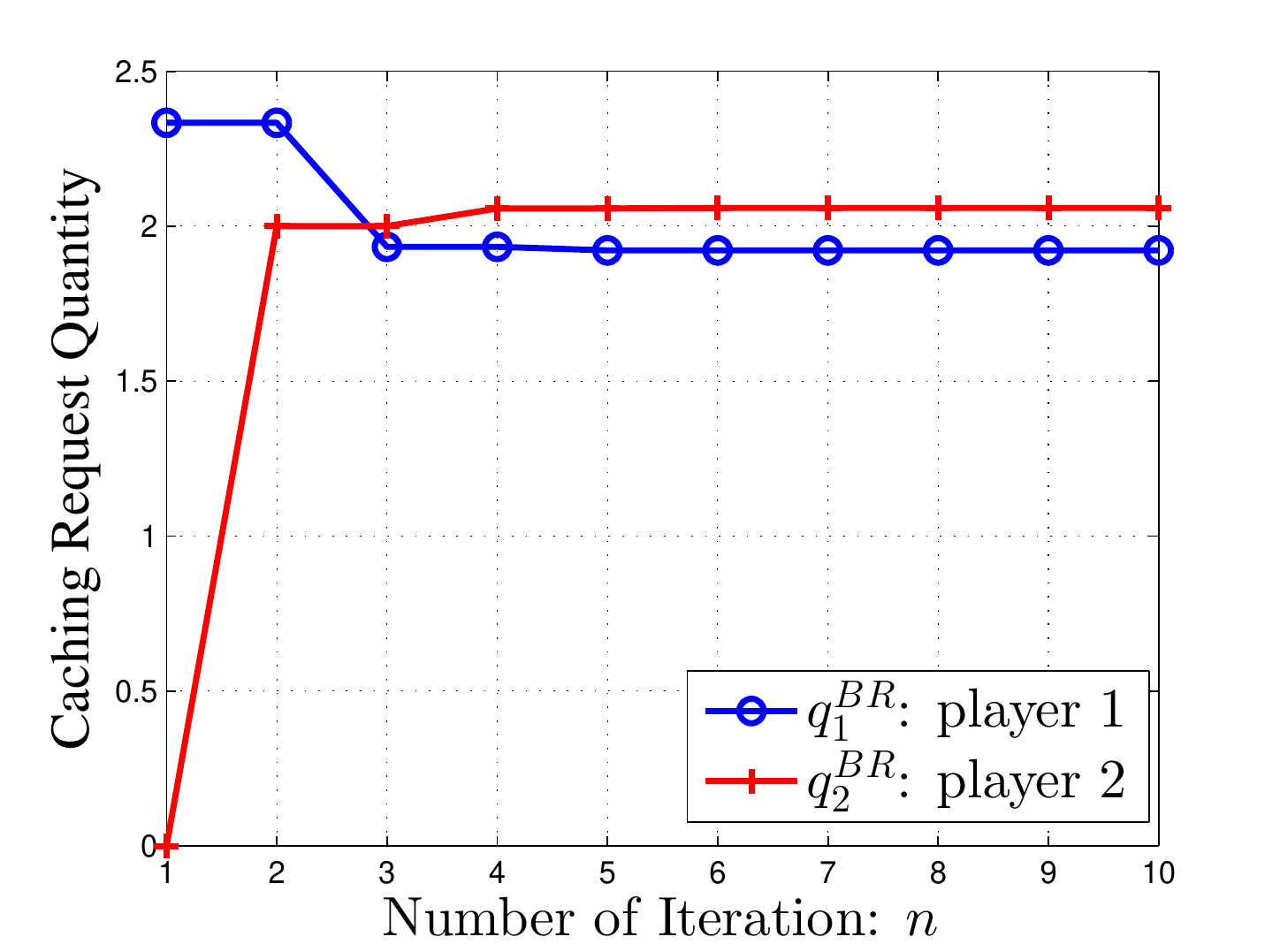}
}
\subfloat[ \label{fig:price}]{%
	\includegraphics[width=0.33\textwidth]{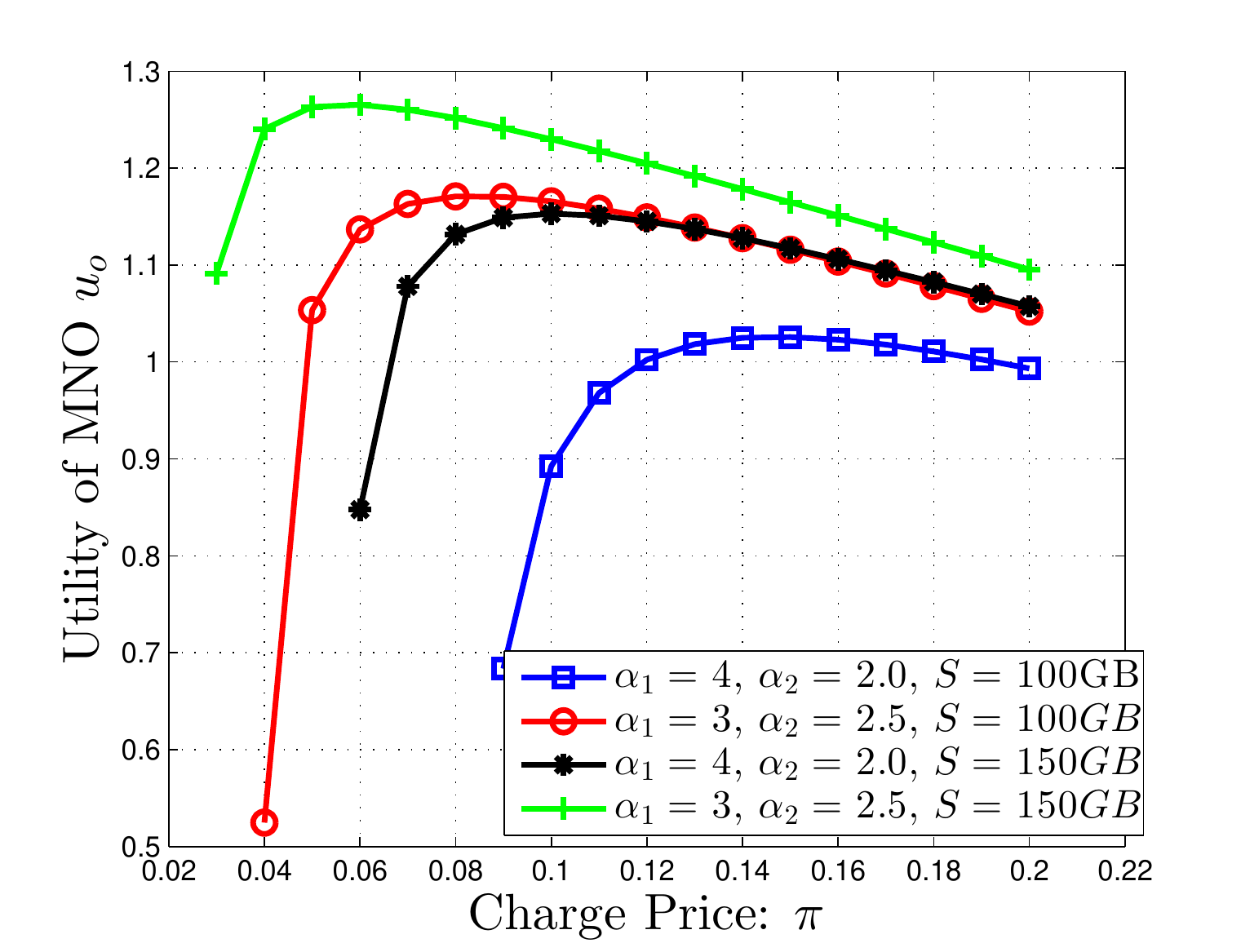}
}
\subfloat[ \label{fig:Ucp}]{%
	\includegraphics[width=0.33\linewidth]{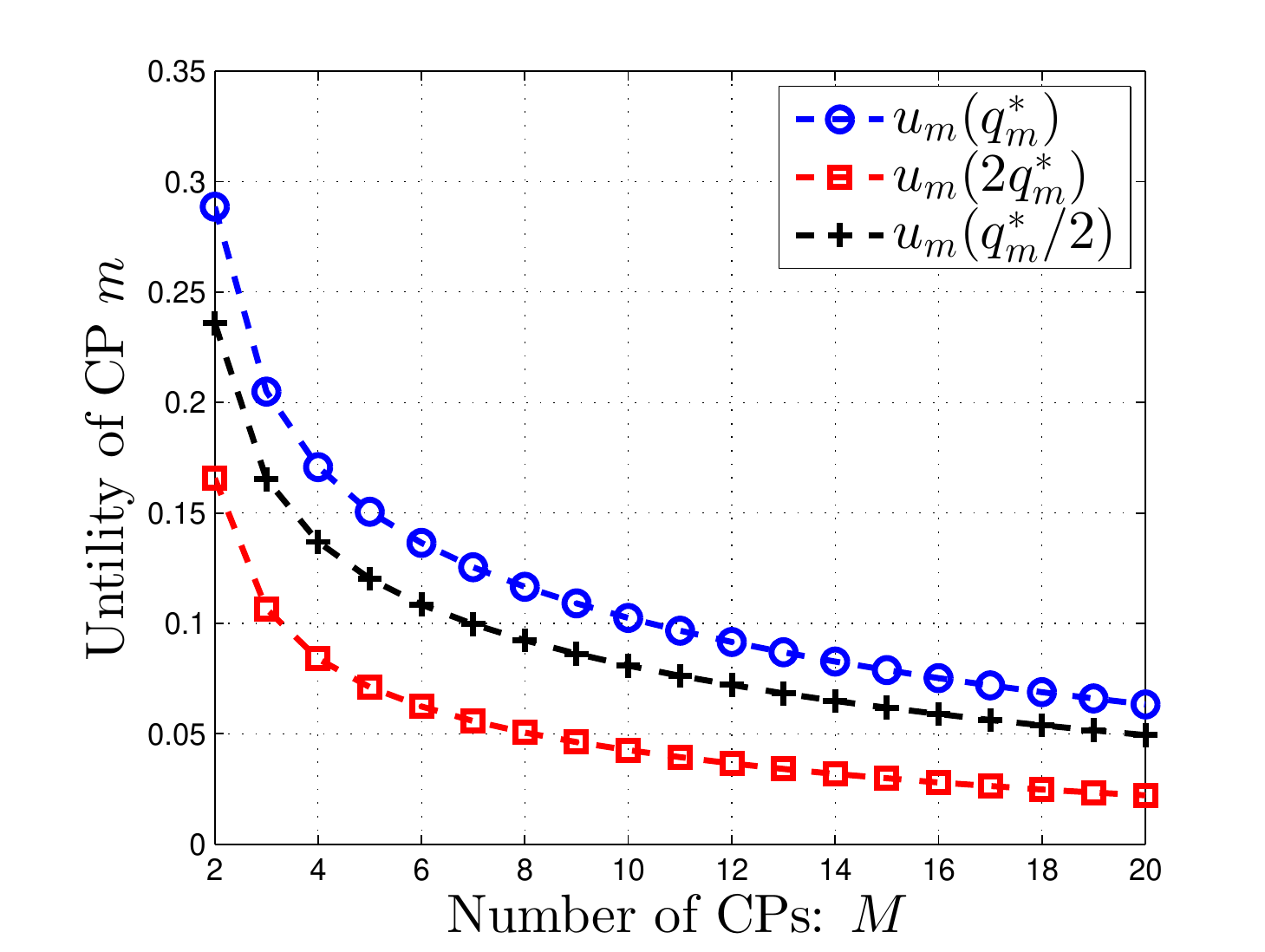}
}
}
\caption{\small Numerical results for a) convergence of the best response for $2$-\ac{CP} case, b) utility function of the \ac{MNO} with respect to charge price $\pi$ for $2$-\ac{CP} case, and c) utility of \ac{CP} with respect to the total number of \glspl{CP} $M$. \vspace{-0.5cm}}
\label{fig:numresults}
\end{figure*}
\vspace{-0.3cm}
\subsection{Example}\label{sec:ex}
\vspace{-0.1cm}
In this part, we show the $2$-\ac{CP} case as an example to illustrate the incentive Stackelberg game on caching. The protocol of the proposed Stackelberg game played between the \ac{MNO} and \glspl{CP} and the non-cooperative game played among the \glspl{CP} are described as follows.
\begin{itemize}
\item The \ac{MNO} predicts the \ac{NE} quantities of the caching requests from all the \glspl{CP} before making its own strategy.
    \begin{eqnarray}
    q_1^{NE} &=& \frac{(\frac{1}{\pi}) (\alpha_1 - 1) \alpha_2}{\alpha_1 \alpha_2 -1},\label{NE1}\\
    q_2^{NE} &=& \frac{(\frac{1}{\pi}) (\alpha_2 - 1) \alpha_1}{\alpha_1 \alpha_2 -1}.\label{NE2}
    \end{eqnarray}
\item The optimal charge price $\pi^{*}$ is the strategy of the \ac{MNO} by maximizing its utility $u_o$.
    \begin{eqnarray}
    \pi^* = \frac{r + \sqrt{\frac{r}{t}}}{r+S},
    \end{eqnarray}
    where $r = \sum_{m=1}^2 \frac{(\alpha_m -1) \alpha_{l\ne m}}{\alpha_1 \alpha_2 -1}$, $t = \sum_{m=1}^2 \frac{(\alpha_m -1) \alpha_{l\ne m}}{\alpha_1 \alpha_2 -1} f(p_m)$ and $f(p_m) = p_m +\frac{\Delta p_m}{2}$.
\item Given the charge price for each file of the caching request, each \ac{CP} chooses a quantity of caching files as its best response.
    \begin{eqnarray}
    q_1^{BR} &=& (\frac{1}{\pi} -1) - \frac{q_2}{\alpha_1}\nonumber\\
    q_2^{BR} &=& (\frac{1}{\pi} -1) - \frac{q_1}{\alpha_2}.\nonumber
    \end{eqnarray}
\item The BR quantities of both \glspl{CP} converge to the \ac{NE} quantities in (\ref{NE1}) and (\ref{NE2}), respectively. The resulting $q_m^{NE}$ maximizes both the utilities of the \ac{MNO} and \glspl{CP}.
\end{itemize}

\vspace{-0.15cm}
\section{Numerical Results}
\label{sec:numerical}
\vspace{-0.1cm}
In order to illustrate the outcome of the proposed game-theoretical cache problem, without loss of generality, we simulate the scenario of one \ac{MNO} and $2$ \glspl{CP} with different parameters. We assume identical file size as $10 GB$ for high quality videos.

Fig. \ref{fig:convergence-3} shows the \ac{BR} dynamic of the non-cooperative game for the $2$-\ac{CP} case. The parameters of the simulations are as follows: $\alpha_1 = 5$, $\alpha_2 = 7$. The charge price provided by the \ac{MNO} is set as $\pi = 0.3$. The initial value of \ac{CP} $2$ is set to be $q_2^{int} = 0$. We can observe that the proposed non-cooperative game converges very rapidly. The convergence is irrespective of the initial points. The convergence values result in the same quantities as the theoretical \ac{NE} solutions given in (\ref{qmNE}) in \textbf{Thereom 1}.

Fig. \ref{fig:price} shows the utility function of the \ac{MNO} with respect to the charge price. We can see that the proposed utility of the \ac{MNO} always admits a global optimum for different sets of parameters. The utilities of the \ac{MNO} with higher caching capacity are higher. This is because more caching requests can be served and less caching cost is spent for the same amount of caching files. The difference of the starting points of the curves are due to the feasible price region in \textbf{Corollary \ref{priceregion}}. The global optimum in Fig. \ref{fig:price} result in the same price derived in (\ref{optiprice}). With the optimum price, the utility of the \ac{MNO} improves even up to $50\%$ than arbitrary chosen prices.

Fig. \ref{fig:Ucp} shows the utility of a single \ac{CP} with respect to the total number of \glspl{CP} $M$ while the quantity of the caching request $q_m$ changes. The charge price is provided as the optimal price $\pi^*$ for different total number of \glspl{CP}. The optimal quantity $q_m^*$ is derived accordingly. The utilities with double quantity $2 q_m^*$ and half quantity $\frac{q_m^*}{2}$ are provided for comparison, respectively. From Fig. \ref{fig:Ucp}, we observe that the more \glspl{CP}, the lower the utility $u_m$ of each \ac{CP}. This is due to the increasing number of CPs which results in the increase of the amount of storage space that is requested by the CPs and thus, a higher price is charged by the MNO. Moreover, the utility of each CP decreases according to (8) when the total storage capacity allocated for the other CPs increases. We can also see that by requesting the optimal caching quantity $q_m^*$, each \ac{CP} achieves $20\%$ higher utility than requesting $\frac{q_m^*}{2}$ and up to $50\%$ than requesting $2q_m^*$. 
\vspace{-0.15cm}
\section{Conclusions}
\label{sec:conclusions}
We have studied a Stackelberg game of proactive edge caching between the leader \ac{MNO} and the followers \glspl{CP}. The best response, the resulting \ac{NE} caching quantities, and the optimum charge price have been derived in closed forms; and the convergence of our proposed incentive mechanism has been validated via numerical studies. These numerical results also showed that both the \ac{MNO} and \glspl{CP} can achieve up to $50\%$ higher utilities in the proposed Stackelberg game. This clearly points out the need of incentive caching mechanisms in $5$G wireless networks.
\bibliographystyle{IEEEtran}
\vspace{-0.2cm}
\bibliography{references}

\begin{thebibliography}{1}
\providecommand{\url}[1]{#1}
\csname url@samestyle\endcsname
\providecommand{\newblock}{\relax}
\providecommand{\bibinfo}[2]{#2}
\providecommand{\BIBentrySTDinterwordspacing}{\spaceskip=0pt\relax}
\providecommand{\BIBentryALTinterwordstretchfactor}{4}
\providecommand{\BIBentryALTinterwordspacing}{\spaceskip=\fontdimen2\font plus
\BIBentryALTinterwordstretchfactor\fontdimen3\font minus
  \fontdimen4\font\relax}
\providecommand{\BIBforeignlanguage}[2]{{%
\expandafter\ifx\csname l@#1\endcsname\relax
\typeout{** WARNING: IEEEtran.bst: No hyphenation pattern has been}%
\typeout{** loaded for the language `#1'. Using the pattern for}%
\typeout{** the default language instead.}%
\else
\language=\csname l@#1\endcsname
\fi
#2}}
\providecommand{\BIBdecl}{\relax}
\BIBdecl

\bibitem{Cisco}
Cisco, ``Cisco visual networking index: Global mobile data traffic forecast
  update, 2015–2020,'' \emph{White Paper}, 2016.

\bibitem{Bastug2014LivingOnTheEdge}
E.~Ba{\c s}tu{\u g}, M.~Bennis, and M.~Debbah, ``Living on the {E}dge: The role
  of proactive caching in {5G} wireless networks,'' \emph{IEEE Communications
  Magazine}, vol.~52, no.~8, pp. 82--89, August 2014.

\bibitem{Pachos2016Technical}
G.~Paschos, E.~Ba{\c s}tu{\u g}, I.~Land, G.~Caire, and M.~Debbah, ``Wireless
  caching: Technical misconceptions and business barriers,'' \emph{arXiv
  preprint arXiv:1602.00173}, 2016.

\bibitem{Poularakis2016Caching}
K.~Poularakis, G.~Iosifidis, A.~Argyriou, I.~Koutsopoulos, and L.~Tassiulas,
  ``Caching and operator cooperation policies for layered video content
  delivery,'' \emph{IEEE International Conference on Computer Communications
  (INFOCOM)}, 2016.

\bibitem{Gregori2016Wireless}
M.~Gregori, J.~G{\'o}mez-Vilardeb{\`o}, J.~Matamoros, and D.~G{\"u}nd{\"u}z,
  ``Wireless content caching for small cell and {D2D} networks,'' \emph{IEEE
  Journal on Selected Areas in Communications}, vol.~PP, no.~99, pp. 1--1,
  2016.

\bibitem{Ghoreishi2016Provisioning}
S.~E. Ghoreishi, V.~Friderikos, D.~Karamshuk, N.~Sastry, and A.~H. Aghvami,
  ``Provisioning cost-effective mobile video caching,'' in \emph{IEEE
  International Conference on Communications (ICC)}, Kuala Lumpur, Malaysia,
  2016.

\bibitem{Peng2016Cache}
X.~Peng, J.~Zhang, S.~Song, and K.~B. Letaief, ``Cache size allocation in
  backhaul limited wireless networks,'' \emph{arXiv preprint arXiv:1602.08728},
  2016.

\bibitem{Alotaibi2015Towards}
F.~Alotaibi, S.~Hosny, J.~Tadrous, H.~E. Gamal, and A.~Eryilmaz, ``Towards a
  marketplace for mobile content: Dynamic pricing and proactive caching,''
  \emph{arXiv preprint arXiv:1511.07573}, 2015.

\bibitem{Chen2016CachingStackelberg}
Z.~Chen, Y.~Liu, B.~Zhou, and M.~Tao, ``Caching {I}ncentive {D}esign in
  {W}ireless {D2D} {N}etworks: {A} {S}tackelberg {G}ame {A}pproach,'' in
  \emph{IEEE International Conference on Communications (ICC)}, Kuala Lumpur,
  Malaysia, 2016.

\end{thebibliography}
\end{document}